\documentclass[11pt, a4paper]{article}
\usepackage{amssymb}
\usepackage{amsthm}
\usepackage{amsfonts}
\usepackage{latexsym}
\usepackage{graphicx,epsfig}
\usepackage{graphicx,subfigure}
\usepackage[table]{xcolor}
\usepackage{authblk}
\usepackage{amsmath}
\usepackage[sort&compress,square,comma,numbers]{natbib}
\usepackage[ruled]{algorithm2e}
\usepackage{setspace}

\usepackage{caption}
\usepackage{array}
\usepackage{multirow}
\usepackage{url}
\usepackage{times}
\makeatletter
\renewcommand{\algorithmcfname}{ALGORITHM}
\SetAlFnt{\small}
\SetAlCapFnt{\small}
\SetAlCapNameFnt{\small}
\SetAlCapHSkip{0pt}
\IncMargin{-\parindent}

\newcommand{\oneandahalfspacing}{\let\CS=\@currsize\renewcommand{\baselinestretch}{1.25}\tiny\CS}

\def\@citex[#1]#2{\if@filesw\immediate\write\@auxout{\string\citation{#2}}\fi
  \def\@citea{}\@cite{\@for\@citeb:=#2\do
    {\@citea\def\@citea{,\linebreak[0]\hskip0pt plus .2em}%
      \@ifundefined{b@\@citeb}%
    {{\bf ?}\@warning{Citation `\@citeb' on page \thepage\space undefined}}%
      \hbox{\csname b@\@citeb\endcsname}}}{#1}}
\newtheorem{theorem}{Theorem}
\newtheorem{property}{Property}

\newtheorem{corollary}{Corollary}
\newtheorem{definition}{Definition}

\newtheorem{rule-def}[theorem]{Rule}
\newtheorem{example}{Example}
\date{}
\newcommand{\be}{\begin{equation}}
\newcommand{\ee}{\end{equation}}
\newcommand{\bea}{\begin{eqnarray}}
\newcommand{\eea}{\end{eqnarray}}

\newcounter{saveeqn}

\begin{document}

\title{\huge {\textbf{An Algorithm for Mining High Utility Closed Itemsets and Generators}}}%
\author[1] {Jayakrushna Sahoo}
\author[2] {Ashok Kumar Das}
\author[3] {A. Goswami}
\affil[1, 3] {\textit{Department of Mathematics, Indian Institute of Technology, Kharagpur 721 302, India} }
\affil[2] {\textit{Center for Security, Theory and Algorithmic Research, International Institute of Information Technology, Hyderabad 500 032, India} }
\maketitle

\begin{abstract}
Traditional association rule mining based on the support-confidence framework provides the objective measure of the rules that are of interest to users. However, it does not reflect the utility of the rules. To extract non-redundant association rules in support-confidence framework frequent closed itemsets and their generators play an important role. To extract non-redundant association rules among high utility itemsets, high utility closed itemsets (HUCI) and their generators should be extracted in order to apply traditional support-confidence framework. However, no efficient method exists at present for mining HUCIs with their generators. This paper addresses this issue. A post-processing algorithm, called the HUCI-Miner, is proposed to mine HUCIs with their generators. The proposed algorithm is implemented using both synthetic and real datasets.\\

{\raggedright{\bf Keywords:}} Data mining, High utility itemset mining, Association rule mining, Condensed representations, non-redundant association rule.
\end{abstract}

\section{Introduction}

Data mining techniques are useful in order to discover efficiently the hidden interesting and useful information from large databases, where the implication of interesting and useful information depends on the problem formulation and the application domain. An important data mining task that has received considerable research attention in recent years is the discovery of association rules from the transactional databases \cite{Agrawal, Park1, HJY1, pjhs, Webb06a, Balcazar:2013, Kryszkiewicz, Zaki3, PTY, Cheng, Yahia, Sahoo2014}. The traditional association rules mining (ARM) techniques depend on support confidence framework in which all items are given same importance by considering the presence of an item within a transaction, but not the profit of item in that transaction. The goal of such techniques is to extract all the frequent itemsets, then generate all the valid association rules $A\rightarrow B$ from frequent itemset $A\cup B$ whose confidence has at least the user defined confidence threshold. In other words, given a subset of the items in an itemset, we need to predict the probability of the purchase of the remaining items in a transactional database. Nevertheless, this support-confidence framework does not provide the semantic measure of the rule but only it provides the statistical measure as the relative importance of items is not considered. However, such measure is not an adequate measure to the decision maker as the itemset cannot be measured in terms of stock, cost or profit, called utility.

In order to address the above shortcoming of support confidence framework, several researchers have focused on weighted association rule \cite{cai, Ramkumar, Wang, Tao, Yun2007, sun, koh, Pears}. In such framework, the weights of items (the importance of items to the user) are considered and it also varies differently in application domains. However, this framework has two pitfalls. Firstly, these schemes still consider the support of an itemset to measure their importance and secondly, these models do not employ the quantities or prices of items purchased. Several researchers have also proposed itemset share measure, which is the fraction of some numerical value in order to overcome these shortcomings \cite{Carter, Hilderman, Hilderman2, Barber2, Barber3, Hilderman1, Li2, Li1}. \citet{Carter} proposed a share-confidence model to discover association rule among numerical attributes which are associated with items in a transaction. Carter et al.'s share-confidence model deals with the amount-share that is a fraction of total weight but not the utility value, such as the net profit, total cost \cite{Geng}. As a result, this model does not accomplish to conventional utility mining \cite{Yao04, Liu:2005, Yao06, Erwin, LiH, LiY, Ahmed, Ahmed:2011, Tseng:2010, Liu:2012, Lan, Tseng:2012, Lin:2013, Song:2013} in which the requirements of decision makers are used to extract the itemsets with high utility, the utility of itemset is no less than the user specified minimum utility threshold, which are composed of weights and purchased quantities. The weight represents the importance of distinct items known as \emph{external utility}, and the purchased quantity in each transaction is known as \emph{internal utility} of the items. The product of external utility with sum total of internal utility of an item is called \emph{utility} of the item. As utility does not satisfy \emph{downward closure property} \cite{Agrawal}, most of the methods proposed in the literature are applied to find the candidate high utility itemsets first and then to identify high utility itemsets by an additional database scan. Some researchers proposed methods to find high utility itemsets without candidate generations \cite{Liu:2012, Fournier:2014}.

The traditional support-confidence model is limited by the number of rules generated that are often not interesting to the user especially when the confidence threshold is very low and some interesting rules are missed when the confidence threshold is very high. In the extracted rule set, most of the rules share the same semantic measure or statistical measure with other rules, and they are called the redundant rules. Henceforth, it limits the usefulness of the rule set for the user to validate and take decisions. To discover the non-redundant association rules in support-confidence framework, several approaches have been proposed in the literature \cite{Kryszkiewicz, PBT, Zaki3, PTY, Szathmary, Cheng, Yahia, Yue, Balcazar:2013, Sahoo2014}. In general, the frequent closed itemsets \cite{PBT} and frequent generators \cite{Zaki3} are used to discover non-redundant association rules. However, these methods are developed for support-confidence framework. Therefore, in this paper we raise an important question as follows. How can we compress the association rules in high utility itemset mining? We aim to answer this question by integrating the concept of frequent closed itemsets and frequent generator to high utility itemset. To represent high utility itemset compactly, \citet{Chan} introduced the concept of utility frequent closed patterns, where the notion of high utility itemset is different from that of the conventional utility mining. Later, \citet{Shie:2010} proposed the maximal high utility itemset, and \citet{Wu:2011, wu:2012} proposed the high utility closed itemset. In these works, there is no algorithm was provided to generate rules. So, if we apply traditional association rule mining, the rules generated from these high utility itemsets would contain redundancy, and also of huge size as they compactly represent high utility itemsets but not the rules.

In order to apply the condensed representation of association rules in support-confidence framework to high utility itemset mining, we first integrate the concept of minimal generators of support-confidence framework to the high utility mining. To extract high utility closed itemsets with their generators simultaneously an algorithm named \emph{HUCI-Miner} (High Utility Closed Itemset-Miner) algorithm has been proposed. From the generated high utility closed itemsets with their generators, various condensed representations of association rules from the literature can be extracted from high utility itemsets. Experiments are carried out to show the compactness of the extracted high utility itemsets.

The rest of the paper is organised as follows. In Section 2, we briefly review the existing works proposed in the literature. In Section 3, we provide some basic preliminaries. Section 4 introduces our proposed HUCI-Miner algorithm. We also discuss the procedure for deriving high utility closed itemsets and generators. In Section 5, we simulate the proposed algorithms method in  real and synthetic datasets. Finally, the paper ends with the concluding remarks in Section 6.

\section{Related work}
In this section, we review the existing methods for high utility itemset mining and generation of non-redundant association rules in support-confidence framework. In the support-confidence framework, the non-redundant association rules are generated from the frequent closed itemsets and their generators. We also review existing works on generation of frequent closed itemsets together with their generators as it is essential for generation of non-redundant association rules.

\subsection{High utility itemset mining}
The traditional association rule mining  methods \cite{Agrawal} are based on support-confidence framework, where all items are considered with the same level of importance. The methods proposed in \cite{Agrawal, Park1, HJY1, pjhs} to extract association rule follow this classical statistical measurement producing the same result on a given minimum support and minimum confidence. The weighted association rule mining (WARM) generalizes the traditional framework by giving importance to items, where importance is given as weights. \citet{Ramkumar} introduced the concept of weighted support of itemsets and weighted association rules on the basis of costs assigned to both items and transactions. Later, considering only the item weights into account, \citet{cai} proposed the weighted support of association rules.  However, the weighted support of the association rules does not satisfy the downward closure property, which results in the performance degradation. In order to overcome such problem, by considering transaction weight, \citet{Tao} provided the concept of weighted downward closure property. Considering both support and weight of itemsets, \citet{Yun2007} then presented a new strategy, called the weighted interesting pattern mining (WIP). \citet{Pears} further proposed a WARM method that automates the process of weight assignment to the items by formulating a linear model.

In WARM framework, note that the quantities of items in transactions are not considered. Considering items' quantities in transactions and their individual importance, high utility itemset mining (HUIM) received a considerable research attention \cite{Yao04, Liu:2005, Yao06, Erwin, LiH, LiY, Ahmed, Ahmed:2011, Tseng:2010, Liu:2012, Lan, Tseng:2012, Lin:2013, Song:2013, Fournier:2014}. Yao et al. \cite{Yao04,Yao06} proposed a mathematical model of utility mining by generalizing the share-confidence model \cite{Barber2}. As utility mining does not fulfill the \emph{downward closure property}, \citet{Liu:2005} proposed the two-phase algorithm that uses the \emph{transaction-weighted downward closure property} to prune the candidate high utility itemsets in the first phase and then all the complete sets of high utility itemsets are obtained in the second phase. To reduce the number of candidate itemsets in the first phase, \citet{LiY} also proposed an isolated items discarding strategy (IIDS) to the level-wise utility mining method.

\citet{Ahmed} proposed a FP-Growth  based algorithm \cite{HJY1} that uses a tree structure, named IHUP-Tree, and efficiently generates the candidate high utility itemsets for incremental and the interactive mining. To further reduce the number of itemsets in the first phase, Tseng et al. \cite{Tseng:2010,Tseng:2012} proposed the tree-based methods, named the UP-Growth and UP-Growth$^+$, which use several strategies to decrease the estimated utility value of an itemset, and as a result, they enhance the performance. To avoid the level wise candidate generation and test strategy, \citet{Song:2013} proposed a concurrent algorithm, called the CHUI-Mine, for mining HUIs from transaction databases using their proposed data structure CHUI-Tree to maintain the information of HUIs. Their proposed algorithm generates the potential high utility itemsets using two concurrent processes: the first process is used for construction and dynamic pruning the tree, and then placing the conditional trees into a buffer, and the second one for reading the conditional pattern list from the buffer and mining HUIs. To speed up the execution and reduce the memory requirement in the mining process,  \citet{Lan} proposed an efficient utility mining approach that adopts a projection-based indexing mechanism that directly generates the required itemsets from the transactions database. \citet{Ahmed:2011} proposed a novel tree-based candidate pruning technique, called the High Utility Candidates Prune (HUC-Prune), for avoiding more database scans and the level-wise candidate generation.

To avoid the computational cost of candidate generation and utility computation, \citet{Liu:2012} then proposed a data structure, named the utility-list, to store both the utility information about an itemset and the heuristic information for pruning the search space. Using the constructed utility-lists from a mined database, they developed an efficient algorithm, called the HUI-Miner, which mines high utility itemsets without candidate generation in a depth-first search manner.  Their algorithm works in a single phase by directly identifying high utility itemsets in an efficient way and it is also scalable. To reduce the cost of join operation in the calculation of the utility-list of an itemset in HUI-Miner, \citet{Fournier:2014} improved the HUI-Miner incorporating with the items co-occurrences strategy (named as FHM) that is about six times faster the HUI-Miner.

\subsection{Closed itemsets with their generators and non-redundant association rule mining}

To generate both frequent closed itemsets (FCI) and generators, \citet{PBT} proposed the CLOSE algorithm that is based on level-wise searching approach with the help of Apriori property. \citet{Szathmary} proposed the ZART algorithm that generates FCIs with their generators in a level-wise manner. They further proposed the Eclat-Z algorithm \cite{Ezart} that mines frequent itemsets in a depth-first way and the FCIs with their generators are identifies in level-wise manner. An effective method, named as Touch \cite{Touch}, was developed by combining the FCI method Charm \cite{Zaki3} and the  frequent generator (FG) mining algorithm, Talky-G \cite{Touch}. The FCIs are mined using Charm and FGs are mined using Talky-G and, then Touch associates the generators to their closed itemsets using a suitable hash function.

\citet{Wu:2011} introduced the closer concept to high utility itemsets. They called the extracted itemsets as closed$^+$ high utility itemsets. On incorporating closure based on support of itemsets, they proved, first mining the set of high utility itemsets and then applying closed constraint produces the same result while mining all the closed itemsets first and then applying the utility constraint. They proposed an effective method named as CHUD (Closed$^+$ High Utility itemset Discovery) for mining closed+ high utility itemsets. Further, they proposed a method called the DAHU (Derive All High Utility itemsets), to recover all high utility itemsets from the set of closed$^+$ high utility itemsets without further accessing the database. In addition, they proposed AprioriHC and AprioriHC-D algorithms \cite{wu:2012} and mentioned that CHUD performs better than AprioriHC and AprioriHC-D. However, no suitable method for high utility closed itemset with their generator is proposed for high utility itemset mining.

To reduce the number of association rules extracted in support confidence-framework, several methods have been developed in the literature \cite{Kryszkiewicz, PBT, Zaki3, PTY, Szathmary, Cheng, Yahia, Yue, Balcazar:2013, Sahoo2014}. \citet{Kryszkiewicz} proposed the representative association rules (RR) with the help of a cover operator that represents a set of association rules. \citet{Zaki3} proposed a method to reduce the number of association rules and the extracted rules, called the general rules, which have shortest antecedent and shortest consequent giving an equivalence class of rules of same support and confidence. \citet{PTY} defined the min-max rules having minimal antecedent and maximal consequent. Their proposed method eliminates the non-redundant rules as min-max exact and min-max approximate rules from the frequent closed itemsets and their generators. Furthermore, to reduce more rules, \citet{Cheng} proposed the concept of $\delta$-tolerance, which is a relaxation on the closure defined on the support of frequent itemset. \citet{Yahia} proposed an informative basis to reduce the number of association rules, which is further efficiently compressed by \citet{Sahoo2014}. \citet{Yue} filtered the min-max rules by defining redundancy and provided the reliable exact basis and reliable approximate basis of the same inference capacity. \citet{Balcazar:2013} further obtained a small and crisp set of association rules by the help of confidence boost of a rule, which eliminates the rules with similar confidence.

\section{Basic preliminaries and problem statement}
In this section, we first discuss some basic preliminaries.

Let $I=\{i_1, i_2, i_3, \ldots, i_m\}$ be a finite set of items, where each item $i_{\ell},~1\leq\ell\leq m$, have an external utility $p_{\ell},~1\leq\ell\leq m$ in the utility table. A subset $X\subseteq{I}$ is called an itemset, if $X$ contains $k$ distinct items $\{i_1, i_2, i_3, \ldots, i_k\}$, where $i_{\ell}\in I,~ 1\leq\ell\leq k$, called a $k$-itemset. Let $\mathcal{D}$ be the task relevant database composed of utility table and the transaction table $T=\{t_1, t_2, t_3, \ldots, t_n\}$, containing a set of $n$ transactions, where each transaction $t_{d}\subseteq{I}, 1 \leq d \leq n$, in the database be associated with a unique identifier, say $t_{id}$. In every transaction $t_d,~1 \leq d \leq n,$ each item $i_{\ell},~1\leq\ell\leq m$ has a non-negative quantity $q(i_{\ell},t_d)$, which represents the purchased quantity known as internal utility of the item $i_{\ell}$ in the transaction $t_d$.

Each itemset $X$ has a statistical measure called the support of $X$, which is defined by the ratio of the number of transactions containing $X$ to the total number of transactions $|\mathcal{D}|$, and denoted by $supp(X)$. In other words, $supp(X)=\frac{\left|\{t \mid t\in \mathcal{D}, X\subseteq{t}\}\right|}{|\mathcal{D}|}$. Let $\mathcal{F}$ be the set of all the itemsets in $\mathcal{D}$ having positive support and $\mathcal{F} = \{X | X\in 2^{\mathcal{I}}, supp(X) > 0\}.$ An association rule is an implication of the form $R: X\rightarrow Y$, where $X, Y \subseteq \mathcal{I}$, $Y \neq \emptyset$, and $X \cap Y = \emptyset$. The itemsets $X$ and $Y$ are called antecedent and consequent of the rule $R$, respectively. Association rules are associated with two statistical measures, which are support and confidence. The support of the rule $R$ is $supp (X \cup Y)$ and the confidence of the rule $R$ is defined by the ratio of the support of $X \cup Y$ to the support of $X$, and denoted by $conf(R)$. Hence, it is clear that $conf(R) = \frac{supp( X \cup Y)} {supp(X)}$.

We use $x_i$ to denote the $i^{th}$ item of an itemset $X$.

\begin{definition}
The utility of an item $i_{\ell}$ in a transaction $t_d$ is denoted by $u(i_{\ell}, t_d)$ and defined by the product of internal utility $q(i_{\ell}, t_d)$ and external utility $p_{\ell}$ of $i_{\ell}$, that is, $u(i_{\ell}, t_d)= p_{\ell}\times q(i_{\ell}, t_d)$.
\end{definition}

\begin{definition}
The utility of an itemset $X$ contained in a transaction $t_d,$ denoted by $u(X, t_d)$ and defined by the sum of utility of every items of $X$ in $t_d$. In other words, $u(X, t_d)=\sum _{i_{\ell}\in X\wedge X\subseteq t_d}u(i_{\ell}, t_d)$.
\end{definition}

\begin{definition}
The utility of an itemset $X$ in $\mathcal{D}$ is denoted by $u(X)$ and defined by the sum of the utilities of $X$ in all the transactions containing $X$ in $\mathcal{D}$, that is, $u(X)=\sum_{X\subseteq t_d\wedge t_d\in \mathcal{D}}u(X,~ t_d)$ $=$ $\sum_{X\subseteq t_d\wedge t_d\in \mathcal{D}}\sum_{i_{\ell}\in X}u(i_{\ell}, ~t_d)$.
\end{definition}

\begin{definition}
An itemset $X$ is called a high utility itemset, if the utility of $X$ has at least the user specified minimum utility
threshold,  $min\_util$. Otherwise, it is called a low utility itemset. Let $\mathcal{H}$ be the complete set of high utility itemsets. Then, $\mathcal{H}=\{ X | X\in \mathcal{F}, u(X)\geq min\_util\}.$
\end{definition}
\begin{example}
Consider again the transaction database in Table \ref{tab:one} with the utility table given in Table \ref{tab:two}. From Table \ref{tab:two}, note that the external utility of item $B$ is $4$ and the internal utility of the item $B$ in the transaction $t_3$ is 4. Thus, the utility of item $B$ in $t_3$ is $u(B, t_3)=p_B\times q(B, t_3)=4\times 4=16$. The utility of the itemset $BF$ in the transaction $t_6$ is $u(BF, t_6)=u(B, t_6)+u(F, t_6)=4\times 1 +1\times 2=6$ and the utility of the itemset $BF$ in $\mathcal{D}$ becomes $u(BF)=u(BF, t_3)+u(BF, t_6)=23$. Note that, if the minimum utility threshold is $20$, the itemset $BF$ is a high utility itemset. Table \ref{tab:three} shows all the high utility itemsets.

\end{example}

\begin{table}[t]
 \begin{minipage}[t]{.45\linewidth}\centering
\centering
\scriptsize
\tabcolsep .1cm
\captionsetup{font=scriptsize}
\caption {An example transaction database $\mathcal{D}$\label{tab:one}}
\begin{tabular}{l l }
        \hline
        $T_{id}$ &       Transaction     \\ \hline
        $t_1$          & $A(4), C(1), E(6), F(2)$ \\
        $t_2$          & $D(1), E(4), F(5)$ \\
        $t_3$          & $B(4), D(1), E(5), F(1)$ \\
        $t_4$          & $D(1), E(2), F(6)$ \\
        $t_5$          & $A(3), C(1), E(1)$ \\
        $t_6$          & $B(1), F(2), H(1)$ \\
        $t_7$          & $D(1), E(1), F(4), G(1), H(1)$ \\
        $t_8$         & $D(7), E(3)$ \\
        $t_9$         & $G(10)$\\ \hline
        \end{tabular}
\end{minipage}
\hspace{-3cm}
\begin{minipage}[t]{.45\linewidth}\centering
\tabcolsep .1cm
\centering
\scriptsize
\captionsetup{font=scriptsize}
\caption {Utility table\label{tab:two}}
\vspace{0.4cm}
\begin{tabular}{ c c c c c c c c c c}
        \hline
        Item & Utility \\\hline
        $A$   & $3$\\
        $B$&$4$\\
        $C$&$5$\\
        $D$& $2$\\
        $E$&$1$\\
        $F$&$1$\\
        $G$&$2$\\
        $H$&$1$ \\ \hline
        \end{tabular}
\end{minipage}
\hspace{-3cm}
\begin{minipage}[t]{.45\linewidth}\centering
\centering
\scriptsize
\captionsetup{font=scriptsize}
\tabcolsep .1pt
\caption {HUIs with minimum utility $20$\label{tab:three}}
\begin{tabular}{c c c c}
        \hline
        Itemset &  utility & Itemset & utility \\ \hline
        $A$     & $21$     &  $B$    & $20$       \\
        $D$     &  $22$    & $G$     &  $22$      \\
        $E$     & $22$     &  $F$    & $20$       \\
        $AC$    & $31$     & $AE$    & $28$       \\
        $BE$    &  $21$    & $DF$    &  $24$      \\
        $BF$    & $23$     & $DE$    & $37$       \\
        $FE$    &  $36$    & $ACE$   & $38$        \\
        $AFE$   &  $20$    & $DFE$   &  $36$      \\
        $BDE$   & $23$     & $BFE$   & $22$       \\
        $ACFE$  & $25$     & $BDFE$  & $24$      \\ \hline
        \end{tabular}
\end{minipage}

\end{table}

\indent We say an association rule in high utility itemset mining is valid, if it satisfies the following two conditions:
\begin{description}
       \item [(i)] The antecedent and itemset formed by combination of antecedent and consequent are high utility itemset.
       \item [(ii)] The confidence is more than or equal to the specified minimum confidence threshold, say $min\_conf.$
\end{description}

Generation of valid utility based association rules from high utility itemsets is relatively straightforward. The rules of the form $R: X\rightarrow Y$ are generated for all high utility itemsets $X$, and $X\cup Y$, for all $X, Y\neq \phi$, and the rule $R$ provides the confidence of the rule having at least $min\_conf$. Since $X\cup Y$ is a high utility itemset, the generated rule is guaranteed to be high utility. To derive all possible valid rules, we need to examine each high utility itemset and repeat the rule generation process as in Apriori algorithm \cite{Agrawal} with utility constraint.

\section{The proposed HUCI-Miner algorithm}
In this section, we first discuss some useful definitions and theorems before describing our proposed algorithm.

The Apriori property used to prune the candidate itemset search space cannot be applied directly to mine high utility itemset, since the utility constraint is neither monotone nor anti-monotone. To reduce the size of search space and enhancing the performance of mining task, \citet{Liu:2005} proposed the concept of transaction-weighted utility, which satisfies the downward closure property and is based on the following definitions.

\begin{definition}
The utility of a transaction $t_{d}$ is denoted by $tu(t_{d})$ and defined by $tu(t_{d}) = $ $u(t_{d},t_{d})$.
\end{definition}

\begin{definition}
The transaction-weighted utility (TWU) of an itemset $X$ in a database $\mathcal{D}$ is denoted by $twu(X)$ and defined by the sum of the utilities of all the transactions containing $X$ in $\mathcal{D}$, where $twu(X) =\sum _{t_{d}\in \mathcal{D} \wedge X\subseteq t_d}tu(t_{d})$.
\end{definition}

\begin{property}
The transaction-weighted utility satisfies the downward closure property. That means for a given itemset $X$, if $twu(X)$ is less than the specified $min\_util$, all supersets of $X$ are not high utility.
\end{property}

\begin{table}[htbp]
\centering
\caption {$TWU$ of each item of database given in Table 1\label{tab:five}}
\begin{tabular}{c c c c c c c c c}
        \hline
       Item & $A$ &  $B$& $C$&  $D$&  $E$&   $F$&  $G$&  $H$ \\
       $TWU$& $40$& $31$& $40$& $72$& $112$& $87$& $30$& $17$  \\ \hline
        \end{tabular}
\end{table}

Table \ref{tab:five} shows the transaction-weighted utility of each item. For example, the transaction utility of $t_1$ is $tu(t_1)=u(A,t_{d})+u(C,t_{d})+u(E,t_{d})+u(F,t_{d})=12+5+6+2=25$. Again, consider the itemset $ACE$ which is in transaction $t_1$ and $t_5$ having the transaction-weighted utility of $ACE$, $twu(ACE)=tu(t_{1})+tu(t_{5})=25+15=40$. If the $min\_util$ is set to $45$, all supersets of $ACE$ are not high utility itemset according to Property 2. For a given itemset, if its transaction-weighted utility has at least $min\_util$, we call the itemset as high transaction-weighted utility itemset (HTWUI).

\begin{definition}
An itemset $Y$ is called the \emph{closure} of an itemset $X$, denoted by $\gamma(X)$, if there does not exist other large superset of $X$ than $Y$, with $supp(X)=supp(Y)$. An itemset $X$ is then called the closed itemset, if $X = \gamma(X)$.
\end{definition}

\begin{property}
For a given itemset $X$,  $twu(X)=twu(\gamma(X))$. In other words, the transaction-weighted utility of an itemset is same as its closure.
\end{property}

\begin{definition}
The local utility value of an item $x_i$ in an itemset $X$, denoted by $luv(x_i, X)$ and defined by the sum of the utility values of the items $x_i$ in all the transactions containing $X$, that is,
$luv(x_i,X)=\sum_{X\subseteq t_d\wedge t_d\in \mathcal{D}}u(x_i, t_d)$.
\end{definition}

\begin{definition}
The local utility value of an itemset $X$ in another itemset $Y$ such that $X\subseteq Y$, denoted by $luv(X,Y)$, is the sum of local utility measure values of each item $x_i\in X$ in itemset $Y$, which is given by
$luv(X,Y)=\sum_{x_i\in X\subseteq Y}luv$ $(x_i,Y).$
\end{definition}

To calculate the local utility value of an itemset $X$ in another itemset $Y$ such that $X\subseteq Y$, an \emph{utility unit array} needs to be attached to each high utility itemset, which is defined as follows.

\begin{definition}\cite{Wu:2011}
The \emph{utility unit array} of an itemset $X=\{i_1,~ i_2, ~i_3, ~\ldots,~ i_k\}$ is denoted by $\mathrm{U}(X)=\{u_1,$ $u_2,$ $u_3,$ $\ldots,$ $u_k\}$, where each $u_\ell$ is $luv(i_\ell,X),~1\leq \ell \leq k.$
\end{definition}

\begin{example}
The \emph{utility unit array}, $\mathrm{U}(X)$ of an itemset $X$ contains the local utility values of the constituent items of $X$. Consider the itemset $ACE$ which appears in transactions $t_1$ and $t_5$ in Table \ref{tab:one}.  The local utility value of the item $A$ in $ACE$ is $luv(A, ACE)= u({A}, t_1) + u({A}, t_5) = 21.$ The utility unit array of $ACE$ is $\mathrm{U}(ACE) = \{21, 10, 7\}$. Further, the local utility value of itemset $AE$ in $ACE$ is $luv(AE, ACE)$ $=$ $luv(A, ACE)$ $+luv(E, ACE)=28$.
\end{example}

\begin{property}\cite{Wu:2011}
For a given itemset $X$ with its \emph{utility unit array} $\mathrm{U}(X)$, the utility of $X$ is defined as $u(X)= \sum_{x_i\in X}$ ${luv (x_i, X)}$.
\end{property}

\begin{property}\cite{Wu:2011}
The \emph{utility} and \emph{utility unit array} of an itemset $X$ can be calculated from the \emph{utility unit array} of its \emph{closure} itemset $\gamma(X)$.
\end{property}

\begin{property}
If an itemset $X$ is a high utility itemset, $\gamma(X)$ is also a high utility itemset. However, the converse is not always true.
\end{property}

\subsection{Integrating the closure property with HUIM}
In this subsection, we discuss how to unify the minimal generator concept of the traditional ARM into high utility itemset mining. In general, the itemsets in a transactional database are not completely independent from other itemsets. A group of itemsets are common to the same set of transactions, and hence, they have the same support. Using the closure operator the itemsets can be grouped into equivalent classes. Two itemsets in a class are called equivalent if they belong to the same set of transactions. A maximal element in a class, is called the closed itemset, which is the closure of other itemsets in that class, and the minimal elements (smallest subsets of the maximal element of the class) are called the generators. All other elements in a class can be derived using the closed itemsets, and the generators and all the elements in a class have the same support. The closed itemsets with their generators lead to the fundamental principle behind the effective construction of non-redundant association rules in the support-confidence framework.

A natural question arises that how to incorporate the similar strategy in HUIM, which means the formation of high utility closed itemset together with their generators those are also high utility itemsets. As suggested by \citet{Wu:2011}, the closure based on utility of itemset does not achieve a high reduction on the number of high utility itemset and they defined the closure on the supports of itemsets.  On incorporating closure based on support of itemsets, they showed that the join order of closed constraint and utility constraint is commutative. This means, first mine the set of high utility itemsets and then apply closed constraint which will produce the same result while mining all the closed itemsets first and then applying the utility constraint. Concluding that any equivalence class constructed among high utility itemsets using closure based on support, the maximal element can be found in both ways. However, later we show that, this commutativeness between utility constraint and closure based on support does not hold for generators. As closed high utility itemsets and high utility closed itemsets are same, so \citet{Wu:2011} defined the closed+ high utility itemset (CHUI).

\begin{definition}
An itemset X is called high utility closed itemset, if $X = \gamma(X)$ and $u(X)\geq min\_util.$
\end{definition}

\subsubsection{Generators with utility constraints}
Following the definition of generators of a closed itemset in traditional ARM, the problem arises how we can incorporate it to high utility itemset mining. This can be done in two ways. First mine all generators based on support constraint and then apply utility constraint. Secondly, first find all high utility itemset and then mine all generator applying support constraint.  We will analyze this joining order between utility constraint and support constraint to extract the generators and conclude that the result of both ordering are different. In first case, we miss the minimal high utility itemsets of an equivalence class and in the later case we do not lose the minimal high utility itemsets. As a consequence, we suggest finding the minimal generators among high utility itemsets. Thus, the utility constraint should be considered first and then support constraint.

For first composition method, we mine generators based on support and then prune itemsets which do not satisfy the utility constraints. Based on composition order, we can define \emph{generator with high utility} as follows.

\begin{definition}[Generator with high utility]
An itemset $X$ is a \emph{generator with high utility}, if there is no proper subset $Z$ of $X$ such that $supp(Z)=supp(X)$. Moreover, $X$ must satisfy the utility constraints.
\end{definition}

Definition 13 simply states that in order to mine \emph{generator with high utility}, using support constraint, the itemsets are first verified whether they are generators or not. Then the itemsets are tested for high utility itemset with utility constraints. So, if an itemset is not a generator, the itemset is not a \emph{generator with high utility}, without verifying with utility constraints. For example, the itemset $AC(2,31)$ (itemset(support,utility)) in the transactional database given in Table\ref{tab:one} with minimum utility constraint $20$, is not a \emph{generator with high utility} without testing with utility constraint as it is not a generator because there is a subset $A(2,21)$ whose support is same as support of $AC(2,31)$. Again, the itemset $AF(1,14)$ is a generator but not a \emph{generator with high utility} as it does not satisfy the utility constraint of minimum utility $20$.

Since the above composition has no restriction on the utility of the subset of generators, so after applying utility constraint some generators are pruned. As all the itemsets of an equivalence class can be generated from both closed itemsets and their generators, if some generators are pruned, some supersets of that itemset may not be generated and result loos of information about high utility itemsets. For example, assume that the minimum utility constraint is $20$ in the transactional database shown in Table \ref{tab:one}. Consider the equivalence class of the itemset $ACFE$ $(1,25)$. The generators of this closed itemset are $AF (1,14)$ and $CF(1,7)$. After applying the utility constraint, note that there is no generator of $ACFE(1,25)$ because both $AF (1,14)$ and $CF(1,7)$ are pruned as their utility is less than $20$. However, from Table \ref{tab:three} we observe that the itemset $AFE(1,20)$ is a high utility itemset belonging to the equivalence class of $ACFE$ $(1,25)$. But in the generation process of generators, the itemset $AFE(1,20)$ is pruned as it is the superset of $AF$ $(1,14)$ with same support. In other words, before $AF$ $(1,14)$ is pruned by the utility constraint, the itemset $AFE(1,20)$ is pruned by itemset $AF(1,14)$ as both have same support and later one is the subset of the former. As a result, the generator itemset of $ACFE(1,25)$ is empty. So, while applying traditional itemset generation procedure from generators of an equivalence class, the itemset $AFE(1,20)$ will not generate. This problem arises because of considering the closure property first and the utility constraint later.

The other way to join the utility constraints and the closure property is that we need to first mine all the itemsets with utility constraint and then apply closure property to compute the generators. From this ordering, we can define \emph{high utility generators} as follows.

\begin{definition}[High utility generators]
An itemset $X$ is a \emph{high utility generator}, if it is a high utility itemset and there exists no proper high utility subset $Z$  such that $supp(Z)=supp(X)$.
\end{definition}

In this approach after extracting all high utility itemsets, the closure property is applied to compute the generators. Again, from these high utility closed itemsets and their generators, all other high utility itemsets of that equivalent class can be generated using the traditional methods of itemset generation. From the analysis, the results of the joining order between the two constraints are different.

\subsubsection{Winding up the discussions}
It is well-defined from the analysis, we loose some minimal high utility itemsets, if we first find generators based on support and then apply utility constraint.  Nevertheless, in Definition 13, we generate the actual high utility generators. As a consequence, the second approach generates high utility minimal generators of an equivalence class in high utility itemset mining. Throughout this paper we apply the second approach and after onwards we mean the generator means the \emph{high utility generator}.

Assume there is a pre-determined total ordering $\Omega$ among the items $I$ in the database $\mathcal{D}$. Accordingly, if item $i$ is occurred before item $j$ in the ordering, we denote this by $i\prec j$. For $\forall j \in Y$, if $i\prec j$, we say $i\prec Y$, where $Y$ is an itemset. Similarly, for itemsets $X$ and $Y$ if $x_i\prec Y_i$ in accordance to the order relation $\Omega$, for $1 \leq i \leq m$, $m =$ $min(|X|,|Y|)$, we say $X \prec Y$. This ordering can be used to enumerate all the itemsets without duplication. Hereinafter, we always consider an itemset as an \textit{ordered set}, in particular, it is a sequence of distinct and increasingly sorted items with respect to the TWU values of items. If the TWU values of two items are equal, they are sorted according to the lexicographic order.

Let $f$ be a function that assign to each itemset $X \in \mathcal{H}$ to the set of all transactions that contain $X$, that is, $f(X)=\{t_d\in T \mid X\subseteq t_d, X\in \mathcal{H}\}$. Clearly, for $X\subset Y$, $f(Y)\subseteq f(X)$. Two itemsets $X, Y \in \mathcal{H}$ are said to be equivalent, denoted by $X\cong Y$, iff $f(X) = f(Y)$. The set of itemsets that are equivalent to an itemset $X$ is denoted by $[X]$ and given by $[X] = \{Y \in \mathcal{H} \mid X\cong Y\}$.

\begin{theorem}
Let $X\in \mathcal{H}$. If $luv(X,Y) = u(X)$, then $Y \in [\gamma(X)]$.
\end{theorem}

\begin{proof}
Since $X\subset Y$, $luv(X,Y) = u(X)$ and $X\in \mathcal{H}$. We then have $Y \in \mathcal{H}$. From Definitions 3, 5 and 6, both $X$ and $Y$ are contained in same transaction. Thus,  $f(X) =$ $f(Y) =$ $f(\gamma(X))$ and  $Y \cong \gamma(X)$. Hence, $Y \in [\gamma(X)]$.
\end{proof}

\begin{corollary}
If $supp(X) = supp(Y)$ and $X \subseteq Y$, then $luv(X,Y) = u(X)$.
\end{corollary}

\begin{corollary}
For a given itemset $X$, if there does not exist any itemset $Y\supset X$ such that $luv(X,Y) = u(X)$, then $X$ is a closed itemset.
\end{corollary}

\begin{definition}
An itemset $X\in[X]$ is called a generator, if $X$ has no proper subset in $[X]$. In other words, it has no proper subset with the same support and it is a high utility itemset.
\end{definition}

\begin{theorem}
Let $X$ be a high utility itemset and $x \in X$. If $X\setminus x$ is a high utility itemset, $X \in [X \setminus x]$ iff $supp(X)=supp(X\setminus x)$.
\end{theorem}

\begin{proof}
Let $x\in X$ and $X\setminus x$ be a high utility itemset. Let $X \in [X\setminus x]$. Then $f(X)=f(X\setminus x)$ means that $supp(X)=supp(X\setminus x).$ The converse follows trivially.
\end{proof}

\begin{theorem}
Let $X$ be a high utility itemset. Then, $X$ is a generator iff $supp(X) \neq min\{supp(X \setminus x): x \in X,$ $(X \setminus x) \in \mathcal{H}\}$.
\end{theorem}

\begin{proof}
Let $X$ be a generator. Let $g$ be a high utility itemset of length $k-1$ with minimum support and a subset of $X$. Then, $g \subset X \Rightarrow f(g)\supseteq f(X)$. If $f(g) = f(X)$, $supp(X) = supp(g)$ and $X$ is not a generator. Moreover, it is not the element with the smallest support, whose closure is $\gamma(X)$. This concludes that $f(g) \supset f(X)$ and hence, $supp(X)\neq min\{supp(X \setminus x): x\in X,$ $(X\setminus x) \in \mathcal{H}\}$. On the other hand, if $supp(X) \neq min\{supp(X \setminus x): x\in X,$ $(X\setminus x) \in \mathcal{H}\}$, $X$ is the smallest element of the closure $\gamma(X)$. Hence, $X$ is a generator.
\end{proof}

\begin{corollary}
Let $X$ be a high utility itemset. If $X$ is not a generator, then $supp(X) = min\{supp(X\setminus x): x\in X,$ $(X\setminus x) \in \mathcal{H}\}$.
\end{corollary}

\begin{property}
For a high utility closed itemset $X$, if $g$ be a generator, then $u(g,X) < u(g',X)$, where $g' \in [X]$ and $g \subset g'$.
\end{property}

\begin{theorem}
Let $X \in \mathcal{H}$. The statement ``If $X$ is a generator, then $\forall Y \in \mathcal{H},$ $Y \subset X$, $Y$ is a generator'' is false in this context.
\end{theorem}

\begin{proof}
This can be proved by giving a counter example. Consider the itemset $AFE$ with support $1$ and utility value $20$. If the minimum utility is set to $20$, this is a generator in the context. However, the subsets $AF$ and $AE$ are not generators. Note that $AF$ is not a high utility itemset, whereas $AE$ is a high utility itemset.
\end{proof}

\noindent\textbf{Rationale:} Theorem 4 states that the subsets of a high utility generator may or may not be a generator. However, in the support-confidence framework the subsets of a generator are generators, and hence, if an itemset is not a generator, its superset is not also a generator. This property is used to pruning the itemsets space to obtain the generators in support-confidence framework. Since this property does not satisfy in the high utility context, this pruning strategy cannot be employed. $AF$ is pruned as it is not a high utility itemset, and $AE$ is pruned as it is not a generator. Thus, by the traditional procedure of a support-confidence framework, the itemset $AFE$ will not be generated. However, it is a minimal itemset in $[AFE]$.

\subsection{Deriving  high utility closed itemsets and generators}
It is observed from the analysis in Section 4.1 that in order to extract all high utility generators, we need to discover the entire space of high utility itemsets. The state-of-the-art algorithm FHM \cite{Fournier:2014} is used to extract all high utility itemsets. The HUCI-Miner algorithm given in Algorithm 1 identifies the high utility closed itemsets and generators among the high utility itemsets. It enables to the efficient generation of the non-redundant association rules among the high utility itemsets. The algorithm outputs the resultant set, $CH$, which contains the high utility closed itemsets $H_k$, where each set $H_k$, $1\leq k\leq max$, has all high utility $k$-itemsets and $max$ is the size of the longest high utility itemset. This algorithm generates the high utility closed itemsets and generators from the high utility itemsets using Theorem 2, 3 and Corollary 3. Note that no additional database scan is required in order to find out the \textit{utility unit array} of each closed itemset, which is used to calculate the local utility value of any subset. By scarifying a little more memory consumption, this can be calculated from the \emph{utlity-list} of the constituent items of an itemset.

The pseudo-code of the HUCI-Miner algorithm, provided in Algorithm 1, is a level-wise procedure. It identifies all the high utility itemsets successively as the high utility closed itemset or generator in each set $H_k$, $1 \leq k \leq max$, contains the high utility itemset of length $k$. It derives the sets $CH_k$, $1 \leq k \leq max$, containing the closed itemsets and their supports, utility values and the corresponding generators. At first, it finds all the high utility itemsets using the FHM algorithm. After this exploration, the algorithm examines each high utility $k$-itemset, $k \geq 2$, which is a generator of a high utility $k$-itemset $(k \geq 2)$ by considering the supports of all its subsets of length $k-1$. The algorithm then verifies if it is a  closed itemset by examining the supports of all its subsets of length $k-1$. Two boolean variables $closed$ and $key$ are used in order to identify whether an itemset is a high utility closed itemset or a generator. If $H_k$ is empty and $H_{k-1}$ is nonempty, by consequence of Property 5 the elements of $H_{k-1}$ are closed and it is performed in Steps 15-16. Conversely, if $H_k$ is nonempty and $H_{k-1}$ is empty, all itemsets in $H_k$ are generators, and no extra step is needed as all itemsets are initially marked as generators.

\IncMargin{1em}
\begin{algorithm}[t]
\SetAlgoNoLine
\LinesNumbered
 \KwIn{$HUI:\{H_1, H_2,\ldots, H_{max}\}$, where $H_k$ is the set of huis of length $k$, $max$ is the size of largest high utility itemset.}
 \KwOut{High utility closed itemsets with their generators.}
      $\mathrm{CH}\leftarrow \emptyset$, $\mathrm{HG}\leftarrow\emptyset$ \tcp{$\mathrm{HG}$: Set of high utility generators }
       $H=FHM(\mathcal{D}, min\_util)$\tcp{$H=\{H_1,H_2,...H_{max}\}$, $H_k$ Set of high utility $k$-itemset}
      \For {each itemset $h\in H_1$}{
      $h.closed\leftarrow true$; $h.key\leftarrow true$\;
      }
      \For {($k=2; k\leq max; k++$)}{
      \If{$H_k\neq \emptyset$}{
        \For {each itemset $h\in H_k$}{
             $h.key\leftarrow true;  h.closed\leftarrow true$\;
              \For {all subsets $h'\in H_{K-1}$ of $h$}{
               \uIf {($supp(h')==supp(h)$)}{
                 $h.key\leftarrow false; h'.closed\leftarrow false$\;
               }
             }
           }
        $CH_{k-1}\leftarrow \{h\in H_{k-1}| h.closed=true\}$\;
        $Get\_generators(CH_{k-1},H_k)$\;
      }
      \Else{
      $CH_{k-1}\leftarrow \{h\in H_{k-1}| h.closed=true\}$\;
        $Get\_generators(CH_{k-1})$\;

      }
      }
       $CH_{k}\leftarrow H_k$\;
       $Get\_generators(CH_{k})$\;
       $\mathrm{CH}\leftarrow \cup_{k} CH_k$ with generators\;
       Calculate \emph{utility unit array} of each closed itemset\;
  End procedure
\caption{$HUCI$-$Miner(\mathcal{D}, min\_util)$ } \label{alg:3}
\end{algorithm}
\DecMargin{1em}

An itemset $c$ is identified as a generator during Steps 8-12 in the algorithm \textit{HUCI-Miner}. If the support of $c$ is same as one of its subsets having length $k-1$ in $H_{k-1}$, then $c$ is not a generator, and conversely, it is not closed. In Steps 15-18, all the closed itemsets of length $k-1$ are added to the set $CH_{k-1}$. Step 21 discovers the closed itemset of the maximum length. In Steps 16, 19 and 22, the $Get\_generators$ procedure is called in order to update the global list of generators and assign the generators to the respective closed itemset.  It takes the set $CH_k$ as input. For each closed itemset $ch \in CH_k$, its proper subsets in the global set of generators $\mathrm{HG}$ are then removed and then added to the the generators of $ch$ (Steps 1-5 in $Get\_generators$ procedure). This procedure updates the the global set of generators $\mathrm{HG}$ by the itemsets, which are not closed but they are generators before the starting of the next iteration. If the set of generators of a given closed itemset is empty, it indicates that the closed itemset is the generator of itself. For example, considering the database given in Table 1 with $min\_util = 20\%$. The HUCIs are $G(),$ $F(),$ $E(),$ $BF(B),$ $DE (D),$  $FE(),$ $ACE (A),$ $DFE (DF),$ $BDFE (BE)$ and $ACFE(AFE)$, where $X(Y)$ means $X$ is the closed itemset and $Y$ is its generator.

\IncMargin{1em}
\begin{procedure}[htbp]
\SetAlgoNoLine
\LinesNumbered
      \KwIn{$CH_{k-1}$: high utility closed itemset of length $(k-1)$}
      \KwOut{Assign the generators to each closed itemsets of $CH_{k-1}$}
       \For {each itemset $ch\in CH_k$}{
             \For {all subsets $c'\in \mathrm{HG}$ of $ch$}{
                add $c'$ in $ch.generator$\;
              }
         }
        $\mathrm{HG}=\mathrm{HG}\cup\{h\in H_k|h.key=true\wedge h.close=false\}$\;
\caption{()\textbf{:}$~Get\_generators(CH_{k-1}, H_k)$ }
\end{procedure}
\DecMargin{1em}

\begin{theorem}
For a given minimum utility threshold, the proposed HUCI-Miner algorithm generates all high utility closed itemsets with their generators correctly.
\end{theorem}

\begin{proof}
The correctness of our proposed HUCI-Miner algorithm is based on Theorems 2 and 3. Theorem 2 determines a high utility itemset $h_{k-1}$, which is a high utility closed itemset, by comparing its support with the supports of the high utility $k$-itemset $h_k$ containing the itemset $h_{k-1}$. Theorem 3 enables if a high utility $k$-itemset $h_k$ is a generator by examining its support with the supports of the high utility $(k-1)$-itemsets, which are included in $h_k$. Since a high utility closed itemset can not be a generator for the large itemsets, they are not in the global list of generators. Again, the closure of a non-generator itemset $c$ is the smallest superset of $c$ in the set of the high utility closed itemsets. This proves that the identification of the generators of a closed itemset is correct.
\end{proof}

\section{Performance evaluation}
In this Section we demonstrate the compactness of high utility itemsets by our proposed HUCI-Miner algorithm using both synthetic (T10I4D100K) and real datasets (Foodmart, Chess, Mushroom, Retail, Chain-store).

\begin{table}[!htb]
\centering
\caption{Characteristics of datasets\label{tab:ten}}
\begin{tabular}{l l l c c c}
        \hline
        Dataset       &     $|T|$        &   $|I|$    & $AvgL$   & $MaxL$  & Type    \\ \hline
        Foodmart      &     $4,141$      &   $1,559$  &  $4.4$   & $14$    & Sparse  \\
        Chess         &     $3,196$      &   $75$     &  $37$    & $37$    & Dense  \\
        Mushroom      &     $8,124$      &   $119$    &  $23$    & $23$    & Dense   \\
        T10I4D100K    &     $100,000$    &   $870$    &  $10.1$  & $29$    & Sparse  \\
        Retail        &     $88,162$     &   $16,470$ &  $10.3$  & $76$    & Sparse   \\
        Chain-store   &     $1,112,949$   &   $46,086$ &  $7.2$   & $170$   & Sparse  \\ \hline
\end{tabular}
\end{table}
The datasets T10I4D100K, Chess, Mushroom and Retail are obtained from frequent itemset mining dataset repository \cite{ms:01}, Chain-store is obtained from NU-MineBench 2.0 \cite{ms:02}, and Foodmart is from Microsoft foodmart 2000 database. Table \ref{tab:ten} shows the characteristics regarding these datasets in terms of the number of transactions ($|T|$), the number of distinct items ($|I|$), the average number of items in a transaction ($AvgL$), the maximum length of transaction $MaxL$, and its type: dense or sparse. Except Chain-store and Foodmart, the other four remaining considered datasets do not provide unit profits of each item (external utility) and item count for each transaction (internal utility). As in \cite{Liu:2005, Ahmed, Ahmed:2011, Tseng:2010, Liu:2012}, we assign in each transaction of T10I4D100K, Mushroom, Chess and Retail, the internal utilities randomly between $1$ and $10$, and the external utilities of each item are randomly generated using a log-normal distribution between the range $0.01$ and $10$.  The FHM implementation is downloaded from the SPMF framework \cite{ms:03}. We implement HUCI-Miner algorithm using the Java programming language and run with the Windows 7 operating system on a machine with CPU clock rate 3.0 GHz and intel core 2 quad processor with 3.18GB of main memory. To enhance the performance of all algorithms, the high utility itemsets are stored in the main memory. In the experiments we have used the minimum utility threshold $min\_util$, as the percentage of total transaction utility values of the database.

\begin{table}[!htb]
\centering
\caption {Number of high utility closed itemset (HUCI) of different datasets \label{tab:eleven}}
\begin{tabular}{l c c c c c}
\hline
Dataset & $min\_util$ in (\%) & $\#$ of HUI & $\#$ of HUCI& $\#$ HG & $\#$ of HUCI+HG \\ \hline
\multirow{4}{*}{Foodmart}
 & 0.07 & 637   & 605  &22  &627\\
 & 0.06 & 1483  &770    &301  &1071\\
 & 0.05 & 6266   &1076  &1573 & 2649\\
 & 0.04 & 20766 &1762   &4686 & 6448\\\hline
 \multirow{4}{*}{Chess}
 & 28 & 1982   &1519  &433 &1952\\
 & 27 & 3874    &2776   &962 &3738\\
 & 26 & 7187  &4910   &1937 &6847\\
 & 25 & 13331   &8700    &3811 &12511\\\hline
\multirow{4}{*}{Mushroom}
 & 8 & 28763   &982  &3253 &4235\\
 &7& 51187    &1470   &5455 &6925\\
 & 6 & 108797  &2216   &10161 &12377\\
 &5 & 222731   &3308    &17303 &20611\\\hline
 \multirow{4}{*}{T10I4D100K}
 & 0.009 & 93348 &66572&9885 &76457\\
 & 0.008 & 111172 &75911&13050 &88961\\
 & 0.007 & 141382 &90187 &18851 &109038\\
 & 0.006 & 192673 & 111544 &28837 &140381\\\hline
 \multirow{4}{*}{Retail}
 & 0.05 & 1876 &1874&2 &1876\\
 & 0.04 & 2810 &2802&8 &2810\\
 & 0.03 & 4848 &4828&20 &4848\\
 & 0.02 & 9341 &9289&52 &9341\\\hline
 \multirow{4}{*}{Chain-store}
 & 0.08 & 118 &118& $-$ & 118\\
 & 0.07 & 151 &151& $-$ &151\\
 & 0.06 & 193 &193& $-$ &193\\
 & 0.05 & 260&260 &  $-$ &260\\ \hline
\end{tabular}
\end{table}

Table \ref{tab:eleven} reports the number of high utility itemsets (HUI), high utility closed itemsets (HUCI), and high utility generators (HG) extracted by HUCI-Miner in various datasets with different $min\_util$ values. The itemset, which is a generator of itself, is not counted in HG. The symbol `$-$' represents that there is no generator except those, who are generators of themselves. From Table \ref{tab:eleven}, we see that the total number of HUCI and HG is less than or equal to the total number of high utility itemsets.  Since all HUIs can be generated from HUCIs with the help of utility unit array, HGs help in finding non-redundant rules in support-confidence framework of high utility mining. Note that our proposed HUCI-Miner algorithm achieves a great reduction in compressing the number of high utility itemsets.

\section{Conclusion}

In this paper, we have first integrated the concept of minimal generator into high utility itemset mining. We have then shown that in order to mine minimal generators in high utility itemset mining, all high utility itemsets need to extract first and then to generate the minimal high utility generators. In order to achieve this, we have proposed a level-wise-search algorithm, HUCI-Miner, that identifies the high utility itemsets, high utility closed itemsets and associates the high utility generators to the corresponding closed itemsets. From the generated high utility closed itemsets with their generators, various condensed representations of association rules from the literature can be extracted from high utility itemsets.

\end{document}